\documentclass[11pt]{article}
\usepackage{graphicx}
\usepackage{amsfonts}
\usepackage{amsmath}
\usepackage{amssymb}
\usepackage{url}
\usepackage{fancyhdr}
\usepackage{indentfirst}
\usepackage{enumerate}
\usepackage{dsfont,footmisc}
\usepackage{booktabs}
\usepackage{subfigure}

\usepackage{setspace}

\usepackage[misc]{ifsym}

\usepackage{amsthm}
\usepackage{color}
\usepackage{natbib}
\usepackage[british]{babel}
\usepackage[colorlinks=true,citecolor=blue]{hyperref}
\usepackage[right,pagewise,mathlines]{lineno}

\usepackage{soul,xcolor}
\setstcolor{red}

\def\d{\mathrm{d}}
\def\laweq{\buildrel \d \over =}

\newcommand{\cov}{\mathrm{Cov}}

\newcommand{\ES}{\mathrm{ES}}
\newcommand{\E}{\mathbb{E}}
\newcommand{\R}{\mathbb{R}}

\newcommand{\Ran}{\mathrm{Ran}}

\newcommand{\p}{\mathbb{P}}

\newcommand{\id}{\mathds{1}}

\newcommand{\X}{\mathbf X}

\renewcommand{\ge}{\geqslant}
\renewcommand{\le}{\leqslant}
\renewcommand{\geq}{\geqslant}
\renewcommand{\leq}{\leqslant}
\renewcommand{\epsilon}{\varepsilon}

\theoremstyle{plain}
\newtheorem{theorem}{Theorem}
\newtheorem{lemma}{Lemma}
\newtheorem{proposition}{Proposition}

\theoremstyle{definition}
\newtheorem{definition}{Definition}
\newtheorem{example}{Example}

\theoremstyle{remark}
\newtheorem{remark}{Remark}

\theoremstyle{definition}

\renewcommand{\cite}{\citet}

\DeclareFontFamily{U}{mathx}{\hyphenchar\font45}
\DeclareFontShape{U}{mathx}{m}{n}{
      <5> <6> <7> <8> <9> <10>
      <10.95> <12> <14.4> <17.28> <20.74> <24.88>
      mathx10
      }{}
\DeclareSymbolFont{mathx}{U}{mathx}{m}{n}
\DeclareMathSymbol{\bigtimes}{1}{mathx}{"91}

\begin{document} 
 
\title{The checkerboard copula and dependence concepts\thanks{Ruodu Wang acknowledges financial support from 
the Natural Sciences and Engineering Research Council of Canada (RGPIN-2024-03728, CRC-2022-00141).
Ruixun Zhang acknowledges research support from the National Key R\&D Program of China (2022YFA1007900) and the National Natural Science Foundation of China (12271013, 72342004).}
}

\author{Liyuan Lin\thanks{Department of Statistics and Actuarial Science, University of Waterloo, Canada.   \texttt{l89lin@uwaterloo.ca}}
\and Ruodu Wang\thanks{Department of Statistics and Actuarial Science, University of Waterloo, Canada.   \texttt{wang@uwaterloo.ca}}
\and Ruixun Zhang\thanks{School of Mathematical Sciences, Center for Statistical Science, and National Engineering Laboratory for Big Data Analysis and Applications, Peking University, China.   \texttt{zhangruixun@pku.edu.cn}}
\and Chaoyi Zhao\thanks{School of Mathematical Sciences, Peking University, China; Sloan School of Management and Laboratory for Financial Engineering, MIT, United States.   \texttt{zhaochaoyi@pku.edu.cn}}
}

\maketitle

\begin{abstract}
We study the problem of choosing the copula when the marginal distributions of a random vector are not all continuous. Inspired by {four} motivating examples including simulation from copulas, stress scenarios, co-risk measures, and dependence measures, we propose to use the checkerboard copula, that is, intuitively, the unique copula with a distribution that is as uniform as possible within regions of flexibility. We show that the checkerboard copula has the largest Shannon entropy, which means that it carries the least information among all possible copulas for a given random vector. Furthermore, the checkerboard copula preserves the dependence information of the original random vector, leading to two applications in the context of diversification penalty and impact portfolios. {The numerical and empirical results illustrate the benefits of using the checkerboard copula in the calculation of co-risk measures. }

\end{abstract}

 \section{Introduction}
 \label{sec:1}

The copula theory has been actively studied over the past few decades with many applications in statistics, finance, engineering, and the natural sciences; for an introduction, see the monographs of \cite{N06} and \cite{J14}.

It is well known through Sklar's theorem (\citet[Theorem 2.10.9]{N06}) that the copula of a random vector is unique if and only if it has continuous marginal distributions. However, when the marginals are non-continuous, the uniqueness of the copula no longer holds. Discrete marginals are common in empirical studies, as the collected data is often discrete. Several works, including \cite{M96}, \cite{C02}, \cite{PSU19}, and \cite{G20}, discuss the dependence structure of discrete data through copulas.
\cite{GN07} discussed difficulties in identifying copulas for discrete distributions.   
The purpose of this paper is to understand whether it is possible to identify a canonical copula for a random vector in some sense if it does not have continuous marginal distributions. 

To answer this question, we seek inspiration from four applications.
Let $\mathbf X = (X_1,\dots,X_d)$ be a $d$-dimensional random vector with $d\geq 2$, which may have non-unique copulas.  Denote by $\mathcal C_{\mathbf X}$ the set of all copulas of $\mathbf X$.  
For a random variable $X$, its probability integral transform $U$ 
is a uniform random variable on $[0,1]$ satisfying $F^{-1}(U)=X$ almost surely (a.s.), where $F$ is the distribution function of $X$ and $F^{-1}$ is the quantile function of $X$.
Let $(U_1,\dots,U_d)$ be any vector of probability integral transforms of $X_1,\dots,X_d$ with a joint distribution $C$; certainly, $C$ is a copula of $\mathbf X$. All random variables live in an atomless probability space $(\Omega,\mathcal F,\p)$.

\begin{enumerate}
 
 \item \textbf{Simulating from the copula of $\X$.} 
 One of the most popular applications of copulas in finance is to model default correlation, as famously done by \cite{L00}; see \cite{MFE15} for discussions.
 In such applications, one needs to simulate from the copula of $\X$, where $\mathbf X$ may have non-continuous marginal distributions (e.g., losses from default events).
 Assume that we can simulate $\X$, and we also have knowledge of all marginal distributions of $\X$. How can we find a reasonable copula $C\in \mathcal C_{\mathbf X}$ to simulate from, that is determined only by $\X$ but not by any particular modeling choices (such as the Gaussian copula)?
 
 \item \textbf{Stressing the distribution of $\X$.}  In sensitivity analysis and risk management, it is often necessary to stress, or distort, the distribution of $\mathbf X$ to obtain post-stress distributions. In the stressing mechanisms studied by \cite{MTW21}, 
 one needs to find a stressed probability measure $Q_1$ by using
 $\d Q_1/\d \p=g(U_1)$ for a non-negative increasing function $g$ with $\int _0^1 g(u) \d u=1$, such as $g(u)=2u$. 
  The simple interpretation of $Q_1$  is to gradually increase the weight of realizations $\omega\in \Omega$ at which $X_1$ is large. Similarly,  one can simultaneously stress all components of $\X$ by considering a measure $Q$ such that $\d Q/\d \p= (1/d) \sum_{i=1}^d g_i(U_i)$ or $\d Q/\d \p=c \prod_{i=1}^d g_i(U_i)$ with a normalizing constant $c>0$ ($c=1$ if $U_1,\dots,U_d$ are independent), where $g_i$ are non-negative increasing functions with $\int _0^1 g_i(u) \d u=1$.
  If we are only interested in the post-stress distribution $\hat F^{Q_1}_1$ of $X_1$ under $Q_1$, the choice of the copula $C\in \mathcal C_{\mathbf X}$ is irrelevant. 
  However, the choice of the copula $C\in \mathcal C_{\mathbf X}$  matters for the distribution $\hat F^Q_i$ of $X_i$ under $Q$, as well as for the distribution $\hat F^{Q_1}_i$ of $X_i$ under $Q_1$. 
  
\item \textbf{Computing a co-risk measure.} 
Co-risk measures (e.g., \cite{AB16})
are calculated for the conditional distribution of a random variable $X_2$ given some event related to $X_1$. 
A classic example is the Marginal Expected Shortfall (Marginal ES) at level $p\in(0,1)$,
which is defined as, assuming that $X_1$ is continuously distributed,
\begin{equation}\label{equ:MES}
\rho(X_2|X_1):=\E[X_2|X_1>F_1^{-1}(p)]=\E[X_2|U_1>p].
\end{equation}
Generally,  $\rho$ is the mean of $X_2$ given a (not necessarily unique) $p$-tail event of $X_1 $ in the sense of \cite{WZ21}. 
 This risk measure $\rho$ does not depend on the choice of $C\in \mathcal C_{\mathbf X}$ if $X_1$ is continuously distributed ($p$-tail event is unique a.s.); however, it may depend on  $C\in \mathcal C_{\mathbf X}$  if $X_1$ has some points of mass.
 Other co-risk measures, such as CoVaR \citep{AB16}, also face the same issue.
 \item \textbf{Maintaining dependence measures.}
Kendall's $\tau$ is a dependence measure defined based on concordance. For a bivariate random vector $(X,Y)$, its Kendall's $\tau$ is defined as 
$$\tau(X, Y)=\p\left((X_1-X_2)(Y_1- Y_2)>0\right)-\p\left((X_1-X_2)(Y_1-Y_2)<0\right),$$
where $(X_1, Y_1)$ and $(X_2, Y_2)$ are two independent copies of $(X,Y)$. If $(X, Y)$ has continuous marginals, we have further
\begin{equation}\label{eq:tau}\tau(X, Y)=4\int_{[0,1]^2} C(\mathbf u)\d C(\mathbf u)-1,
\end{equation}
where $C$ is the copula for $(X, Y)$. 
If $(X,Y)$ has non-continuous marginals, \eqref{eq:tau} does not hold for every $C\in \mathcal C_{(X,Y)}$.
It is natural to ask which copula $C\in \mathcal{C}_{(X, Y)}$ conveys the property in the case of continuous marginals. We can also consider similar applications for other concordance-based dependence measures such as Spearman’s $\rho$.
 \end{enumerate}

All of the above contexts point to the question of choosing a good copula $C\in \mathcal C_{\mathbf X}$.
{\cite{N04} discusses similar applications for the choice of copulas. The main idea of \cite{N04} is to extend the subcopula, which is the part of the copula that uniquely determined by joint distribution, to capture the dependence of the original random vector analogous to the case with continuous marginals. 
In this paper, we address this problem from the view of probability integral transformation.}  
We first offer a new characterization of all copulas of a given random vector {by constructing all probability integral transformations} in Section \ref{sec:2} in Theorem \ref{lem:1}.
In Section \ref{sec:3}, we give some intuitive and heuristic arguments for the questions above, leading to
the proposal of using the checkerboard copula, that is, the unique copula of $\mathbf X$ that is as uniform as possible in regions where the copulas of $\mathbf X$ are not uniquely determined, formally defined in Definition \ref{def:0}. {The checkerboard copula is the same as the standard extension proposed by \cite{N04}, which has been shown to preserve quadrant dependence, tail dependence, and weak convergence results of the joint distribution.}
Although the arguments in Section \ref{sec:3} are heuristic, the use of the checkerboard copula indeed has a theoretical justification, which we present in Section \ref{sec:ent}. The checkerboard copula has the maximum Shannon entropy among all possible copulas of $\mathbf X$, as shown in Theorem \ref{th:entropy}.
In Section \ref{sec:5}, we show in Theorem \ref{th:gen} that the checkerboard copula preserves various dependence concepts that are satisfied by $\mathbf X$. This result is intuitive, but the proof requires serious technical analysis. 
We discuss two applications of our results in diversification penalty and induced order statistics in Section \ref{sec:app}. Section \ref{sec:application_corisk} uses numerical and empirical experiments to demonstrate that the checkerboard copula is a convenient and natural choice that can produce reliable results. Section \ref{sec:conclusion} concludes the paper. 

 \section{Copulas for a discrete random vector}\label{sec:2}

Let $d\ge 2$ be an integer and $[d]=\{1,\dots,d\}$.
All inequalities are
interpreted component-wise when applied to vectors.
All random variables live in an atomless probability space $(\Omega,\mathcal F,\p)$. Let $\X=(X_1,\dots,X_d)$ be a $d$-dimensional random vector, $F_1,\dots,F_d$ be the marginal distributions of $\X$, and $\mathrm{Ran}(F_i)$ be the range of $F_i$ for $i \in [d]$. By Sklar's theorem, the copula of $\X$ is uniquely determined on $\mathrm{Ran}(F_1) \times \dots \times \mathrm{Ran}(F_d)$ but undetermined in other regions. Therefore, when the marginal distribution $F_i$ is not continuous for some $i \in [d]$, the copula of $\X$ may not be unique. In this section, we give a concrete representation for any copulas of $\X$. 
 
We start with the observation that, if a random variable $X$ is continuously distributed, the random variable 
 $$
 U_X:=F_X(X) 
 $$
 will be uniformly distributed over $[0,1]$, where $F_X$ is the cumulative distribution function of $X$. More generally, regardless of whether $X$ is continuously distributed, we can define its probability integral transform 
 \begin{align}\label{eq:u1}
 U_X:=F_X(X-) + V_X(F_X(X)- F_X(X-)) ,
 \end{align}
 where $F_X(x-)=\lim_{y\uparrow x} F_X(x)=\p(X<x)$ for $x\in \R$ and $V_X\sim \mathrm{U}[0,1]$ is independent of $X$, assumed to exist.\footnote{This assumption is safe as we are interested in distributional properties, and we can extend the probability space to include such independent $V_X$, if necessary.} 
The  probability integral transform  $U_X$     satisfies $U_X\sim \mathrm{U}[0,1]$ and $F_X^{-1} (U_X)=X$ a.s.~(see e.g., \cite[Proposition 1.3]{R13}). Therefore, the probability integral transform \eqref{eq:u1} converts any random variable $X$ to a $\mathrm{U}[0,1]$ distributed random variable $U_X$ using $V_X$. 

We extend this idea to the case of a random vector $\X$. Let $\mathbf V=(V_1,\dots,V_d)$ be a random vector with $\mathrm{U}[0,1]$ marginals such that $V_i$ is independent of $X_i$ for each $i\in [d]$.  Denote the set of such $\mathbf V$ by $\mathcal V_\X$.
Similar to \eqref{eq:u1}, let us define the probability integral transform for $\X=(X_1,\dots,X_d)$:
\begin{align}\label{eq:us}
 U_{i}:=F_{i}(X_i-) + V_i(F_{i}(X_i)- F_i(X_i-)),~~~i\in[d].
 \end{align}
It immediately follows that $U_i\sim  \mathrm{U}[0,1]$ and $F_i^{-1}(U_i)=X_i$ a.s. Therefore, $\mathbf U=(U_1,\dots,U_d)$ is a random vector with uniform marginals.  This transformation comes from randomized hypothesis tests \citep[Section 5.3]{F67} and has been applied in various contexts; see, e.g., \cite{MS75}, \cite{N07}, \cite{R81,R09,R13}, and  \cite{F17}.

Let $C_{\mathbf X}^{\mathbf V}$ be the copula of $\mathbf U$. Because $F_i^{-1}(U_i)=X_i$ a.s.~for each $i \in [d]$, we have 
\begin{align*}
 C_{\mathbf X}^{\mathbf V} (F_1(x_1),\dots,F_d(x_d)) &=\p(U_1\le F_1(x_1),\dots,U_d\le F_d(x_d)) 
 = \p (X_1\le x_1,\dots,X_d\le x_d)
\end{align*}
for any $ (x_1,\dots,x_d)\in  \R^d$. 
Hence, $C_{\mathbf X}^{\mathbf V}$  is a copula of $\X$.

According to \eqref{eq:us}, the copula $C_{\mathbf X}^{\mathbf V}$ is determined by the joint distribution of $(\mathbf X,\mathbf V)$. In particular, the copula $ C_{\mathbf X}^{\mathbf V} $ does not depend on the choice of $V_i$ for $i$ such that $X_i$ is continuously distributed because, for these $i$, $U_i$ in \eqref{eq:us} is a.s.~equal to $F_i(X_i)$. While for $i$ such that $X_i$ is discrete, $V_i$ does have an impact on the copula $ C_{\mathbf X}^{\mathbf V} $.

In general, the choice of $\mathbf V \in \mathcal V_\X$ for constructing the copula $C_{\mathbf X}^{\mathbf V}$ may not be unique. This is because $\mathcal V_\X$ allows two types of dependence that might be present in the construction of  $\mathbf V$: First, the components of $\mathbf V$ may be mutually dependent. Second, $V_i$ may depend on $X_j$ for $i\ne j$. Naturally,  a different choice of $\mathbf V \in \mathcal V_\X$ often leads to a different copula $ C_{\mathbf X}^{\mathbf V} $; see the following example.

 \begin{example}\label{ex:1}
Assume that $d=2$, $X_1$ is a constant, and $X_2$ is continuously distributed.  It is well known that any copula is a copula of $\X$ in this case. For instance, by choosing $V_1$ to be independent of $X_2$, $C_{\X}^{\mathbf V}$ is the independence copula, and by choosing $V_1=F_2(X_2)$, $C_{\X}^{\mathbf V}$ is the comonotonic copula.
 \end{example}

The following result says that all copulas of $X$ can be realized by some $C_{\X}^{\mathbf V}$. Hence, \eqref{eq:us} gives a stochastic representation for any copula of $\X$. The representation is quite intuitive, but we did not find it in the literature, so we provide a self-contained proof. 
 \begin{theorem}\label{lem:1}
 Let $\X$ be a random vector 
 such that there exists a continuously distributed random variable independent of  $\X$.
 A copula $C$ is a copula of $\X$ if and only if $C=C_{\X}^{\mathbf V}$ for some $\mathbf V \in \mathcal V_\X$.
 \end{theorem}  
 
 \begin{proof}
 We have seen that $C_{\X}^{\mathbf V}$ is a copula of $\X$. It suffices to show the ``only if" statement. 
 Let $C$ be a copula of $\X$,   take $\mathbf U'=(U_1',\dots,U_d')\sim C$, and 
 write $\mathbf X'= \left(F^{-1}_1(U_1'),\dots,F^{-1}_d(U_d')\right)$. 
  Because $C$ is a copula of $\X$, for $\mathbf x =(x_1,\dots,x_d)\in \R^d$, we have
\begin{align*}
 \p(\X\le \mathbf x)  =  C\left(F_1(x_1),\dots,F_d(x_d)\right) &= \p\left(U'_1\le F_1(x_1),\dots,U'_d\le F_d(x_d)\right)  \\&  = \p\left(F^{-1}_1(U_1')\le x_1,\dots,F^{-1}_d(U_d')\le x_d\right) =\p(\X'\le \mathbf x).
 \end{align*}
 Hence, $\mathbf X \laweq \X'$. 
Take  $\mathbf U^*=(U_1^*,\dots,U_d^*) $ such that 
$(\X,\mathbf U^*)\laweq (\X',\mathbf U')$, and we then have $\mathbf X = \left(F^{-1}_1(U_1^*),\dots,F^{-1}_d(U_d^*)\right)$ a.s. 
Furthermore, take $  V'\sim \mathrm U[0,1]$ which is independent of $(\X,\mathbf U^*)$. 
The existence of $\mathbf U^*$ and $V'$ 
is guaranteed by the assumption of the existence of a continuously distributed random variable independent of $\X$.
For $i\in[d]$, let $\mathbf V=(V_1,\dots,V_d)$ be given by $$V_i= \frac{U_i^*-F_i(X_i-)}{F_i(X_i)- F_i(X_i-)} \id_{\{F_i(X_i)>F_i(X_i-) \}} +  V' \id_{\{F_i(X_i)=F_i(X_i-) \}}.$$
Fix $i\in[d]$ below. 
Let $D_i$ be the set of   discontinuity points of $F_i$.
Note that for $x\in D_i$, we have $$\p\left( U_i^*\in [F_i(x-), F_i(x) ]  \big|X_i=x\right) =1 \mbox{~~and~~} \p\left( U_i^*\in [F_i(x-), F_i(x) ]  \big|X_i \ne x\right) =0.$$
Because $U_i^*$ is uniformly distributed over $[0,1]$, 
  $U_i^*$ is uniform on $[F_i(x-), F_i(x) ] $  conditional on   $X_i=x\in D_i$. Thus,
\begin{align*}
\p(U_i^* \le u | X_i=x )   =    \frac{u-F_i(x-)}{F_i(x)-F_i(x-)} ,~~~~u\in [F_i^{-1}(x-), F_i^{-1}(x) ].
\end{align*} 
Therefore, for $u\in [0,1]$,   
\begin{align*}
\p(V_i \le u | X_i) & = \p\left(  \frac{U_i^*-F_i(X_i-)}{F_i(X_i)-F_i(X_i-)} \le u \Big| X_i \right) \id_{\{X_i\in D_i \}} + \p(V'\le u)   \id_{\{X_i \not \in D_i\}}
\\& = \p\left(  U_i^* \le  u (F_i(X_i)-F_i(X_i-)) +F_i(X_i-)  \big|X_i\right) \id_{\{X_i\in D_i \}}   +  u  \id_{\{X_i\not \in D_i \}}  
\\& = u  \id_{\{X_i\in D_i \}}   +  u  \id_{\{X_i \not \in D_i \}} 
 = u.
\end{align*}
Hence, $V_i$ follows $\mathrm{U}[0,1]$ and is independent of $X_i$.
Note that, by the construction, $U_i^*$, $V_i$, and $X_i$ satisfy $U_{i}^*=F_{i}(X_i-) + V_i(F_{i}(X_i)- F_i(X_i-))$ a.s., 
and hence $\mathbf U^*\sim C^{\mathbf V}_\X$.   
This shows $C= C^{\mathbf V}_\X$.
 \end{proof}

 Theorem \ref{lem:1} implies $\mathcal C_\X = \{C_\X^{\mathbf V}: \mathbf V\in \mathcal V_\X\}$, providing a characterization of copulas in a stochastic form.
Note that $\mathcal C_\X$ is a singleton if and only if all marginal distributions of $\X$, $F_1, \dots,F_d$, are continuous functions. {The ``if'' direction of Theorem \ref{lem:1} in the case $d=2$ is shown by \citet[Proposition 4]{N07}. The characterization of copulas in an analytical form is provided by \cite{DCDS17}.}

\section{Motivating arguments for the checkerboard copula}
\label{sec:3}

Theorem \ref{lem:1} gives the entire class of copulas for $\X$. 
 We now consider which $\mathbf V \in \mathcal V_\X$ can answer the {four} motivating questions in Section \ref{sec:1}, which all point to the same unique choice of $\mathbf V\in \mathcal V_\X$.

  \begin{enumerate}
 
 \item \textbf{Simulating from the copula of $\X$.}  A natural approach to simulating from the copula of $\X$ with some atoms in the marginal distributions is by first simulating a pair of $(\X,\mathbf V)$, and then applying the probability integral transform using \eqref{eq:us}.
 Theorem \ref{lem:1} shows that all copulas of $\mathbf X$ can be simulated this way. 
 For this purpose, the simplest and most natural choice of $\mathbf V$ is $\mathbf V\sim \mathrm{U}\left([0,1]^d\right)$ which is independent of $\mathbf X$. In fact, we could not think of an argument against the use of this particular $\mathbf V$ in the context of simulation. 
 
 \item \textbf{Stressing the distribution of $\X$.} 
 To understand how the choice of $\mathbf V$ affects the stressed distribution of $X_2$, we look at the simple example  
 in Example \ref{ex:1} with $g(u)=2u$. Choosing $V_1$ independent of $X_2$ would lead to $\hat F^{Q_1}_2=F_2$, whereas choosing $V_1=F_2(X_2)$ would lead to $\hat F^{Q_1}_2=(F_2)^2$.
 Because we are interested in the effect of stressing $X_1$ on $X_2$, and $X_1$ is a constant in this example, it is natural to choose a $V_1$ that affects the distribution of $X_2$ minimally, which is achieved when $V_1$ is independent of $X_2$.
  Translating this argument into the general $d$-dimensional setting suggests choosing $\mathbf V\sim \mathrm{U}\left([0,1]^d\right)$ independent of $\mathbf X$.
  
\item \textbf{Computing a co-risk measure.} 
 To understand how the choice of $\mathbf V$ affects the value of the co-risk measure, we again look at Example \ref{ex:1}.  We have $\rho(X_2|X_1)=\E[X_2]$ if $V_1$ is independent of $X_2$, and $\rho(X_2|X_1)=\ES_p(X_2)$ if $V_1=F_2(X_2)$, where $\ES_p(X_2)=\E[X_2|U_2>p]$ is the Expected Shortfall of $X_2$ at level $p$.
 The interpretation of $\rho$ as the mean of $X_2$ on a tail event of $X_1$  suggests that it is natural to choose $V_1$ independent of $X_2$, because $X_1$ is a constant and its tail event should not affect $X_2$.

{ \item \textbf{Maintaining dependence measures.} It has been shown in  \cite{DL05}, \cite{N07}, and \cite{GN07} that \eqref{eq:tau} holds for the copula introduced by $\mathbf V \sim \mathrm{U}([0,1]^2)$. 
This property leads to extensions of dependence measures in the multivariate case based on this copula; see, e.g., \cite{MQ10} and \cite{GNR13}.
}

 \end{enumerate} 
 
In all the considerations above, $\mathbf V\sim \mathrm{U}\left([0,1]^d\right)$ independent of $\mathbf X$ appears to be a good choice.
 Let us denote this by $\mathbf V^{\perp}_\X$ and the corresponding copula by $C_{\X}^{\perp}$, where $\perp$ reflects that independence is used twice to construct $\mathbf V$ (within components of $\mathbf V$ and between $\mathbf V$ and $\X$).  
From the four motivating examples above,  the choice of the particular copula $C_{\X}^{\perp}$ is natural and has several unique features. 
This choice has been known as the checkerboard copula. 

\begin{definition}\label{def:0}
The copula $C_{\X}^{\perp}$ is called the \emph{checkerboard copula} of $\X$.
\end{definition}

The copula $C_{\X}^{\perp}$ is also called  the multilinear extension copula of $\X$; see \cite{GNR17} for its properties, its empirical process, and a history. One notable property is that $X_1,\dots,X_d$ are independent if and only if $C_\X^{\perp}$ is the independence copula. 

The rest of the paper focuses on the properties and applications of the checkerboard copula.

\section{Entropy maximization}
\label{sec:ent}
Given the natural choice of $C_{\X}^{\perp}$ in the applications in Section \ref{sec:3}, it should have some unique properties within the class $\mathcal C_\X$. The applications seem to suggest that $C_{\X}^{\perp}$ relies less on external information compared to other choices of $\mathbf V$.
 Such consideration is typically studied via entropy. 
Indeed, as argued by \cite{J57}, the maximum-entropy distribution should be the only unbiased choice given available information. 
If a copula $C $ has a density function $c$, then its Shannon (differential) entropy is defined as 
$$
H(C)=-\int_{[0,1]^d} c(\mathbf u) \log  c(\mathbf u)  \d \mathbf u.
$$
One problem with the above formulation is that a copula $C$ often does not have a density. 
We set $H(C)=- \infty$ if $C$ does not have a density, which is intuitive and can be seen as a limiting case; 
{see \cite{KPRH16} for a discussion on the definition of entropy for singular distributions. 
\begin{remark}
By definition, $H(C)=-D_{\rm KL}(P_C \Vert P_L)$, where $D_{\rm KL}(P_C \Vert P_L)$ is the KL divergence between the probability measure $P_C$ with distribution function $C$ and the Lebesgue measure $P_L$ on $[0,1]^d$. 
Since the KL divergence quantifies the similarity between $P_C$ and $P_L$ via the likelihood ratio $\d P_C/\d P_L$, 
being singular is the extreme form of non-similarity in terms of likelihood ratio. 
Therefore, it is natural to set $D_{\rm KL}(P_C \Vert P_L) $ as  $\infty$  whenever $P_C$ is not absolutely continuous with respect to $P_L$.
This is a standard approach in the literature on KL divergence and differential entropy; 
see e.g., \cite{C75} and \cite{CT91}. Hence, to keep the same intuition and consistency with the KL divergence, we set $H(C)=-\infty$ when $C$ does not have density. 
\end{remark}}
However, even the checkerboard copula $C^{\perp}_{\X}$ may not have a density if the distribution of $\X$ has some singular continuous part. This issue may be solved by considering other measures of information, but for now, let us stick to the Shannon entropy, which is the most popular notion in information theory. We would like to compare $H( C^{\perp}_{\X})$ with $H(C)$ for $C \in\mathcal C_{\mathbf X} $, 
or equivalently, $H(C^{\mathbf V}_{\X})$ for other choices of $\mathbf V\in \mathcal V_\X$. 
The main result of this section is to show that $H(C_{\mathbf X}^\perp)$ has the largest entropy among all other choices.

\begin{theorem}
\label{th:entropy}
For   $C\in \mathcal C_\X$, 
 we have $H(C^{\perp}_\X) \ge H(C )$.
\end{theorem}

The proof of Theorem \ref{th:entropy} essentially boils down to showing the following lemma, which states that the density of the checkerboard copula can be expressed as the conditional expectation for the density of other possible copulas in $\mathcal{C}_{\X}$.
From this lemma and Jensen's inequality, Theorem \ref{th:entropy} follows. 

\begin{lemma}\label{lem:condition}
For $C\in \mathcal{C}_\X$, if the density $c$ of $C$ exists, then the density $c^\perp$ of $C^\perp_\X$ exists. Moreover, $c^{\perp}(\mathbf U)=\E[c(\mathbf U)|\hat \X]$, where  $\mathbf U =(U_1, \dots, U_d)\sim \mathrm{U}\left([0,1]^d\right)$, $\mathbf {\hat \X}=\left(F_1^{-1}(U_1), \dots, F_d^{-1}( U_d)\right)$, and  $F_1, \dots, F_d$ are the  marginals of $\X.$
\end{lemma}
\begin{proof}
Since $\E[c(\mathbf U)| \hat \X]$ is $ \sigma(\hat \X)$-measurable, there exists a function $f: \R^d \to [0,1]$ such that $f(\hat \X)=\E[c(\mathbf U)|\hat \X]$ (in the almost sure sense).
Let $c^\perp$ be a function $[0,1]^d \to [0,1]$ defined as $ c^\perp(\mathbf u)=f\left(F^{-1}_1(u_1), \dots, F^{-1}_d(u_d)\right)$ for any $\mathbf u=(u_1, \dots, u_d) \in [0,1]^d$.
We claim that $c^\perp$ is the density of $C^\perp_{\mathbf X}$. This claim implies that $c^\perp$ exists and $c^{\perp}(\mathbf U) = \E[c(\mathbf U)|\hat \X]$. 

To prove this claim, let $\mathbf U^c \sim C$, $\mathbf U^{\perp}\sim C_{\X}^{\perp}$, and $R=\bigtimes_{i=1}^d \mathrm{Ran}(F_i)$. 
We first show that $\int_A c^\perp (\mathbf u)\d \mathbf u=\p(\mathbf U^\perp \in A)$ for the following two types of the set $A$.

\begin{enumerate}[(i)]
\item Let $A=\bigtimes_{i=1}^d [0,a_i]$ with $\mathbf a=(a_1, \dots, a_d) \in R$. We have $\id_{\{\mathbf U \le \mathbf a\}}=\id_{\left\{\hat \X\le \left(F_1^{-1}(a_1), \dots, F_d^{-1}(a_d) \right) \right\}}$.
Therefore,
\begin{align*}
\int_{A} c^\perp(\mathbf u) \d \mathbf u
&=\int_{\bigtimes_{i=1}^d[0, a_i]} f\left(F_1^{-1}(u_1), \dots, F_d^{-1}(u_d)\right) \d u_1 \cdots \d u_d\\
&=\E\left[f(\hat \X)\id_{\{\mathbf U \le \mathbf a\}}\right]\\
&=\E\left[\E[c(\mathbf U)|\hat \X]\id_{\left\{\hat \X\le \left(F_1^{-1}(a_1), \dots, F_d^{-1}(a_d) \right) \right\}}\right]\\
&=\E\left[c(\mathbf U)\id_{\{ \mathbf U \le  \mathbf a\}}\right]=\p(\mathbf U^c \le  \mathbf a)=\p(\mathbf U^\perp \le   \mathbf a),
\end{align*}
where the last equality holds because
\begin{align*}
\p\left(\mathbf U^c \le \mathbf a\right)=\p\left(\mathbf X \le \left( F_1^{-1}(a_1), \dots, F_d^{-1}(a_d)\right)\right)=\p(\mathbf U^\perp \le  \mathbf a).
\end{align*}
This further implies that $\int_{A} c^\perp(\mathbf u) \d \mathbf u=\p(\mathbf U^\perp\in A)$ for any $A=\bigtimes_{i=1}^d A_i$ such that $A_i \in \{[0,a_i]: a_i \in \mathrm{Ran}(F_i)\} \cup \{(F_i(x_i-), F_i(x_i)]: x_i \text{ is a discontinuity point of }F_i \}$ for $i\in [d]$.
\item 
Let $A=(\bigtimes_{i=1}^k[0, a_i]) \times (\bigtimes_{j=k+1}^d (s_j, t_j])$ with $k \in \{0,1,\dots,d\}$ such that $a_i\in \mathrm{Ran}(F_i)$ for $i\in [k]$ and $(s_j, t_j] \cap \mathrm{Ran}(F_j)=\varnothing$ for $j\in [d]\setminus[k]$. For $j \in [d]\setminus[k]$, denote by $x_j=F_j^{-1}(s_j)$, and thus $(s_j, t_j] \subseteq (F_j(x_j-), F_j(x_j))$. By the definition of $c^\perp$, for fixed $u_i\in [0, a_i]$ and $i\in [k]$, $c^\perp(u_1, \dots, u_k, v_{k+1}, \dots, v_d)$ is a constant for all $(v_{k+1}, \dots, v_d) \in \bigtimes_{j=k+1}^d(F_j(x_j-), F_j(x_j))$. Therefore, let $B=(\bigtimes_{i=1}^k[0, a_i]) \times (\bigtimes_{j=k+1}^d (F_j(x_j-),F_j(x_j)))$, we have
\begin{equation*}
\int_{A} c^\perp(\mathbf u) \d \mathbf u= \left( \prod_{j=k+1}^d\frac{t_j-s_j}{F_j(x_j)-F_j(x_j-)} \right)\int_{B} c^\perp(\mathbf u) \d \mathbf u.
\end{equation*}
Let  $\mathbf V=(V_1, \dots, V_d) \sim \mathrm{U}\left([0,1]^d\right)$ be independent of $\mathbf X$, and for $j\in [d]\setminus [k]$, denote by $s'_j={(s_j-F_j(x_j-))}/{(F_j(x_j)-F_j(x_j-))}$ and $t'_j={(t_j-F_j(x_j-))}/{(F_j(x_j)-F_j(x_j-))}$. Hence,
$$
\prod_{j=k+1}^d\frac{t_j-s_j}{F_j(x_j)-F_j(x_j-)}  = \p\left( V_j \in (s'_j, t'_j] \mbox{ for  all}~j \in[d]\setminus [k] \right).
$$
In addition, by (i), we can get 
\begin{align*}
\int_{B} c^\perp(\mathbf u) \d \mathbf u=\p(\mathbf U^\perp \in B) =\p\left(X_i\le F^{-1}(a_i), ~ X_j=x_j \mbox{ for  all}~ i\in [k],~j \in [d]\setminus[k] \right). 
\end{align*}
Therefore, 
\begin{align*}
 \int_{A} c^\perp(\mathbf u) \d \mathbf u 
&= \p\left( \bigcap_{j \in[d]\setminus [k]}\{V_j \in (s'_j, t'_j]\}  \right)  \p\left(
\bigcap_{i\in [k], j \in [d]\setminus[k]}\{
X_i\le F^{-1}(a_i),~ X_j=x_j\}\right)\\
&=\p\left(\bigcap_{i\in [k], j \in [d]\setminus[k]}\left\{X_i\le F^{-1}(a_i),~ X_j=x_j ,~ V_i\in [0,1],~V_j \in (s'_j, t'_j]\right\}\right)\\
&=\p\left(\mathbf U^\perp \in \left(\bigtimes_{i=1}^k[0, a_i] \right)\times \left(\bigtimes_{j=k+1}^d (s_j, t_j]\right)\right)=\p ( \mathbf U^\perp \in A) .
\end{align*}

By the same argument, we have $\int_{A} c^\perp(\mathbf u) \d \mathbf u=\p(\mathbf U^\perp\in A)$ for any $A=\bigtimes_{i=1}^d A_i$ such that $A_i \in \{[0,a_i]: a_i \in \mathrm{Ran}(F_i)\} \cup \{(s_i, t_i]: (s_i, t_i] \cap \mathrm{Ran}(F_i)=\varnothing\}$ for $i\in [d]$.
\end{enumerate}

For any $\mathbf a=(a_1,\dots,a_d) \in [0,1]^d$, the region $\bigtimes_{i=1}^d [0, a_i]$ can always be represented by an at most countable disjoint union of regions studied in (i) and (ii). Hence,  we can obtain 
$$\int_{\bigtimes_{i=1}^d [0, a_i]} c^\perp(\mathbf u) \d \mathbf u=\p(\mathbf U^\perp \le \mathbf a).$$
This proves our claim that $c^\perp$ is the density of $C^\perp_{\mathbf X}$.
\end{proof}

\begin{proof}[Proof of Theorem \ref{th:entropy}]
If  $H(C )=-\infty$, there is nothing to show.
Hence, it suffices to consider the case that $C$ has a density, which we denote by $c$.  
By Lemma \ref{lem:condition}, we have $c^{\perp}(\mathbf U) = \E[c(\mathbf U)|\hat \X]$  where $c^\perp$ is the density of $C^\perp_\X$,  $\mathbf U =(U_1, \dots, U_d)\sim \mathrm{U}\left([0,1]^d\right)$, and $\mathbf {\hat \X}=\left(F_1^{-1}(U_1), \dots, F_d^{-1}( U_d)\right)$ with  $F_1, \dots, F_d$ as the  marginals of $\X.$
Define a function $g(x)=x\log x$ for $x\in (0,\infty)$. It is clear that $g$ is convex.
By the fact that $\E[c(\mathbf U)|\hat \X]=c^{\perp}(\mathbf U)$ and Jensen's inequality, we have 
\begin{align*}
H(C^{\perp}_\X)=-\E[g(c^\perp(\mathbf{U}))]=-\E[g(\E[c(\mathbf{U})|\hat \X])]\ge -\E[\E[g(c(\mathbf{U}))|\hat \X]]=-\E[g(c(\mathbf{U}))]=H(C ).
\end{align*}
Thus, $H(C^{\perp}_\X) \ge H(C )$ for all $C\in \mathcal C_\X$.
\end{proof}

{
Theorem \ref{th:entropy} demonstrates that the entropy of $C_{\mathbf X}$ cannot be greater than $C^\perp_{\mathbf X}$. This result is related to those of \cite{PHB12} and \cite{KSU20}, and the difference is that, in Theorem \ref{th:entropy}, we fix a joint (possibly discrete) distribution and seek to find the copula consistent with this distribution that maximizes the entropy, which is the checkerboard copula. In contrast,
\cite{PHB12} and \cite{KSU20} do not fix a joint distribution. Instead, they search for a checkerboard copula that maximizes the entropy subject to matching either a correlation coefficient or the distribution of the sum of random variables. Therefore, the problems they address are different from our Theorem \ref{th:entropy}.
}

{
When $C_{\mathbf X}$ and $C_{\mathbf X}^\perp$ contain singular components, by definition, $H(C_{\mathbf X})=H(C^\perp_{\mathbf X})=-\infty$. In this case, the next proposition shows that, the entropy for the absolutely continuous part of $C_{\mathbf X}^\perp$ is still greater than that for the absolutely continuous part of $C_{\mathbf X}$.

\begin{proposition}
Assume $C_{\mathbf X} \in \mathcal{C}_{\mathbf X}$ such that $C_{\mathbf X}=\lambda G_{\rm A}+(1-\lambda) G_{\rm S}$, where $\lambda \in [0,1]$, $G_{\rm A}$ is an absolutely continuous distribution function, and $G_{\rm S}$ is a singular distribution function. Then, there exists an absolutely continuous distribution function $ G^\perp_{\rm A}$ and a  distribution function $ G^\perp_{\rm S}$ such that $C_{\mathbf X}^\perp=\lambda G^\perp_{\rm A}+(1-\lambda)  G^\perp_{\rm S}$ and $H(G^\perp_{\rm A})\ge H( G_{\rm A})$.
\end{proposition}
\begin{proof}
Let $F_1,\dots, F_d$  be the marginal distributions of $\mathbf X$.
For $x\in [0,1]$, let $l_i(x)=\sup\{y: y \in \mathrm{Ran}(F_i), ~y\le x\}$ and $u_i(x)=\inf\{y: y \in \mathrm{Ran}(F_i),~ y\ge x\}$.
 Define two distribution functions $G^\perp_{\mathrm A}$
and $G^\perp_{\mathrm S}$, which are linear interpolations of $G_{\mathrm A}$
and $G_{\mathrm S}$ from $\bigtimes_{i=1}^d \mathrm{Ran}(F_i) $ to $[0,1]^d$: For $\mathbf x=(x_1, \dots, x_d) \in [0,1]^d$, 
$$G_{\rm A}^\perp(\mathbf x)=\left\{\begin{aligned}
&G_{\rm A}(\mathbf x),~~ &\mathbf x \in \bigtimes_{i=1}^d \mathrm{Ran}(F_i),\\
&\sum_{ \substack{y_i \in \{l_i(x_i), u_i(x_i)\} \\i\in[d]}} \prod_{j=1}^d \beta_j(x_j,y_j)
G_{\rm A}(y_1, \dots, y_d),~~&\mathbf x \in [0,1]^d \setminus \bigtimes_{i=1}^d \mathrm{Ran}(F_i),
\end{aligned}\right.$$
 and 
$$G_{\rm S}^\perp(\mathbf x)=\left\{\begin{aligned}
&G_{\rm S}(\mathbf x),~~ &\mathbf x \in \bigtimes_{i=1}^d \mathrm{Ran}(F_i),\\
&\sum_{\substack{y_i \in \{l_i(x_i), u_i(x_i)\}\\i\in [d]}} \prod_{j=1}^d \beta_j(x_j,y_j)
G_{\rm S}(y_1, \dots, y_d),~~&\mathbf x \in [0,1]^d \setminus \bigtimes_{i=1}^d \mathrm{Ran}(F_i),
\end{aligned}\right.$$
where $$
\beta_i(x,y)=\frac{u_i(x)-x}{u_i(x)-l_i(x)} \id_{\{y=l_i(x)\}}+\frac{x-l_i(x)}{u_i(x)-l_i(x)}\left(1-\id_{\{y=l_i(x)\}}\right)
$$
for $x,y \in [0,1]$ and $i\in [d]$ with the convention $0/0=1$. 
 Note that $\beta_i(x,y)$ is linear in $x$ on each segment of $[0,1]\setminus \mathrm{Ran}(F_i)$. It is clear that $G_{\rm A}^\perp$ is continuous.

Let $g_{\rm A}$  be the density of $G_{\rm A}$ and  $g$ be the derivative of $G^\perp_{\rm A}$, respectively. 
For $\mathbf x\in \bigtimes_{i=1}^d \Ran(F_i)$,
we have 
$$g(\mathbf x)=\frac{\partial^d G_{\rm A}^\perp(\mathbf x)}{\partial x_1\cdots \partial x_d}=\frac{\partial^d G_{\rm A}(\mathbf x)}{\partial x_1\cdots \partial x_d}=g_{\rm A}(\mathbf x).$$
For $\mathbf x =(x_1, \dots, x_{d-1},x_d)$ such that $x_i\in \Ran(F_i)$ for $i\in [d-1]$ and $x_d\notin \Ran(F_d)$,
we have $$G_{\rm A}^\perp(\mathbf x)=\frac{u_d(x_d)-x_d}{u_d(x_d)-l_d(x_d)}G_{\rm A}(x_1, \dots, x_{d-1}, l_d(x_d))+\frac{x_d-l_d(x_d)}{u_d(x_d)-l_d(x_d)}G_{\rm A}(x_1, \dots, x_{d-1}, u_d(x_d)).$$
Hence,
\begin{align*}g(\mathbf x)&=\frac{\partial^d G_{\rm A}^\perp(\mathbf x)}{\partial x_1\cdots\partial x_d}\\
&=\frac{\partial^{d-1}[ G_{\rm A}(x_1, \dots, x_{d-1}, u_d(x_d))-G_{\rm A}(x_1, \dots, x_{d-1}, l_d(x_d))]}{\partial x_1\cdots \partial x_{d-1}}=\int_{l_d(x_d)}^{u_d(x_d)}g_{\rm A}(x_1, \dots, x_{d-1}, y) \d y.\end{align*}
Let $k\in [d]$. Similarly,
for $\mathbf x=(x_1, \dots, x_d)$ such that $x_i\in \Ran(F_i)$
 for $i\in [k-1]$ (with $[0]=\varnothing$) and $x_i\notin \Ran(F_i)$
 for $i\in [d]\setminus[k-1]$, we have 
 $$g(\mathbf x)=\int_{l_{k}(x_{k})}^{u_{k}(x_{k})}\dots\int_{l_d(x_d)}^{u_d(x_d)}g_{\rm A}(x_1, \dots, x_{k-1},y_{k}, \dots, y_d) \d y_{k} \cdots\d y_d.$$
Note that $u_i(x_i)=F_i(F_i^{-1}(x_i))$ and $l_i(x_i)=F_i(F_i^{-1}(x_i)-)$ for $i\in [d]$. From the discussion above, we can see that $g(\mathbf x)$ is a constant in $\bigtimes_{i=1}^d [F_i(y_i-), F_i(y_i)]$ for all $\mathbf y=(y_1, \dots, y_d) \in \R^d$.
Let $\mathbf U=(U_1, \dots, U_d)\sim \mathrm{U}([0,1]^d)$ and ${\hat \X}=\left(F_1^{-1}(U_1), \dots, F_d^{-1}( U_d)\right)$. Thus,
$g(\mathbf U)=\E[g_{\rm A}(\mathbf U)|{\hat \X}]$.
 Hence $g$ is  Lebesgue integrable and $G_{\rm A}^\perp$ is an absolutely continuous distribution function with density function $g_{\rm A}^\perp=g$.

Next, we show that $C_{\mathbf X}^\perp=\lambda G^\perp_{\rm A}+(1-\lambda)  G^\perp_{\rm S}$ and $H(G^\perp_{\rm A})\ge H( G_{\rm A})$.
For  $\mathbf x\in \bigtimes_{i=1}^d \mathrm{Ran}(F_i)$, 
$$\lambda G_{\rm A}^\perp(\mathbf x)+(1-\lambda)G_{\rm S}^\perp(\mathbf x)=\lambda G_{\rm A}(\mathbf x)+(1-\lambda)G_{\rm S}(\mathbf x)=C_{\mathbf X}(\mathbf x)=C_{\mathbf X}^\perp(\mathbf x).$$
For  $\mathbf x\in [0,1]^d \setminus \bigtimes_{i=1}^d \mathrm{Ran}(F_i)$, 
\begin{align*}\lambda G_{\rm A}^\perp(\mathbf x)+(1-\lambda)G_{\rm S}^\perp(\mathbf x)&= \sum_{\substack{y_i \in \{l_i(x_i), u_i(x_i)\}\\i\in [d]}} \prod_{j=1}^d \beta(x_j,y_j) \left(\lambda G_{\rm A}(y_1, \dots, y_d)+(1-\lambda)G_{\rm S}(y_1, \dots, y_d)\right)\\
&=\sum_{\substack{y_i \in \{l_i(x_i), u_i(x_i)\}\\i\in [d]}} \prod_{j=1}^d \beta(x_j,y_j) C_{\mathbf X}(y_1, \dots, y_d)=C_{\mathbf X}^\perp(\mathbf x).
\end{align*}
Let $h(x)=-x \log x$ for $x \in (0, \infty)$. It is clear that $h$ is a concave function.
By Jensen’s inequality, we have
\begin{align*}
H(G^\perp_{\rm A})=\E[h(g_{\rm A}^\perp(\mathbf U))]
=\E[h(\E[g_{\rm A}(\mathbf U)|{\hat \X}])]\ge \E[\E[h(g_{\rm A}(\mathbf U))|{\hat \X}]]=\E[h(g_{\rm A}(\mathbf U))]=H(G_{\rm A}).
\end{align*}
This completes the proof.
\end{proof}}

 \section{Checkerboard copula and dependence concepts}
 \label{sec:5}

In this section, we study how the checkerboard copula preserves dependence concepts.  
This question is motivated by a problem raised in the context of diversification in \cite{CEW24}, which we describe in Section \ref{sec:cew}.

\subsection{Dependence concepts}
\label{sec:51}
We first define several notions of positive dependence, introduced and studied by \cite{L66}, 
\cite{EPW67}, and \cite{BY01},
and the corresponding notions of negative dependence, introduced and studied by
\cite{L66}, \cite{AS81}, \cite{BSS82, BSS85},  \cite{JP83}, and \cite{CEW24a}. 

In what follows, for $i\in [d]$ and an $d$-dimensional random vector $\mathbf X=(X_1,\dots,X_d)$,   write  $\mathbf X_{-i}=(X_1,\dots,X_{i-1}, X_{i+1},\dots,X_d)$, and  for $A,B\subseteq [d]$, write $\mathbf X_A=(X_k)_{k\in A}$ and $\mathbf X_B=(X_k)_{k\in B}$. 
A set $S\subseteq \R^d$ is  \emph{decreasing} 
if $\mathbf x\in S$ implies $\mathbf y\in S$ for all $\mathbf y\le \mathbf x$. 

\begin{definition}\label{def:1} 
A random vector $\mathbf X$ is  
\begin{enumerate}[(i)]
\item  
\begin{enumerate}[(a)]
\item 
\emph{positively associated} (PA)   if for every pair of  subsets $A,B$ of $[d]$ and any functions  $f$ and $g$ both increasing or decreasing coordinatewise, provided the covariance below exists,
$$
\cov(f(\mathbf X_A),g(\mathbf X_B))\ge 0;
$$  
\item  \emph{negatively associated} (NA)   if for every pair of disjoint subsets $A,B$ of $[d]$ and any functions  $f$ and $g$ both increasing or decreasing coordinatewise, provided the covariance below exists,
\begin{equation*}\label{eq:NA}
\cov(f(\mathbf X_A),g(\mathbf X_B))\le 0;
\end{equation*}
\end{enumerate}
\item  
\begin{enumerate}[(a)]
\item 
\emph{positively regression dependent} (PRD) if for every $i\in[d]$,  the random variable $\E[g(\mathbf X_{-i})| X_i]$
is an increasing function of $X_i$ for any  coordinatewise increasing function $g$ such that the conditional expectation exists; 
\item \emph{negatively regression dependent} (NRD) if for every $i\in[d]$,  the random variable $\E[g(\mathbf X_{-i})| X_i]$
is a decreasing function of $X_i$ for any  coordinatewise increasing function $g$ such that the conditional expectation exists;
\end{enumerate}
 \item  
\begin{enumerate}[(a)]
\item 
 \emph{weakly positively associated} (WPA) if  
 for any $i\in[d]$,  decreasing set $S  \subseteq  \R^{d-1}$, and $x\in \R$ with $\p(X_i\le x)>0$,  
$$
 \p(\mathbf X_{-i} \in S  \mid  X_i\le  x) \ge \p(\mathbf X_{-i}\in S); $$
 \item  \emph{weakly negatively associated} (WNA) if  
 for any $i\in[d]$,  decreasing set $S  \subseteq  \R^{d-1}$, and $x\in \R$ with $\p(X_i\le x)>0$,  
\begin{equation*}\label{eq:WNA}
 \p(\mathbf X_{-i} \in S  \mid  X_i\le  x) \le \p(\mathbf X_{-i}\in S);  \end{equation*}  
 \end{enumerate} 
 \item 
 \begin{enumerate}[(a)]
\item 
 \emph{positively orthant dependent} (POD) if for all $\mathbf x =(x_1,\dots,x_d) \in\R^d$,  $\p(\mathbf X\le \mathbf x)\ge \prod_{i=1}^d\p(X_i\le x_i)$ and $\p(\mathbf X >  \mathbf x)\ge \prod_{i=1}^d\p(X_i > x_i)$;
 \item  \emph{negatively orthant dependent} (NOD) if for all $\mathbf x =(x_1,\dots,x_d) \in\R^d$,  $\p(\mathbf X\le \mathbf x)\le \prod_{i=1}^d\p(X_i\le x_i)$ and $\p(\mathbf X >  \mathbf x)\le \prod_{i=1}^d\p(X_i > x_i)$.
\end{enumerate} 
\end{enumerate}  
Moreover, we say that a distribution or a copula is  PA, PRD, WPA, POD, NA, NRD, WNA, or NOD
if the corresponding random vector is. 
\end{definition}
 Note that the definition of PA does not require $A$ and $B$ to be disjoint, whereas the definition of NA requires this.
 
The relationship between the above notions is summarized below (see e.g., \cite{CEW24a}).
$$
\mbox{PA} \Longrightarrow \mbox{WPA};~~~
\mbox{PRD} \Longrightarrow \mbox{WPA};~~~ \mbox{WPA}  \Longrightarrow \mbox{POD};
$$ 
$$
\mbox{NA} \Longrightarrow \mbox{WNA};~~~
\mbox{NRD} \Longrightarrow \mbox{WNA};~~~ \mbox{WNA}  \Longrightarrow \mbox{NOD}.
$$
Within the class of multivariate normal distributions, the four concepts of positive dependence are equivalent, and each is equivalent to having nonnegative bivariate correlation coefficients;
similarly, the four concepts of negative dependence are equivalent, and each is equivalent to having nonpositive bivariate correlation coefficients.

In the sequel, we use $\mathfrak D$ to represent one of the following: PA, PRD, WPA, POD, NA, NRD, WNA, or NOD. 
Our question is whether these properties are properties purely based on copulas.
It turns out that the checkerboard copula can help answer this question. 

\subsection{The checkerboard copula preserves dependence}

We first present a self-consistency property of those negative dependence concepts in the spirit of \cite[Property P$_6$]{JP83} for NA. 
\begin{lemma}\label{lem:2}
If $f_1, \dots, f_d$ are increasing functions and $\mathbf X$ satisfies $\mathfrak D$,  then $(f_1(X_1), \dots, f_d(X_d))$   also satisfies $\mathfrak D$.
\end{lemma}
\begin{proof}
We only show the result for the concepts of negative dependence, as the case of positive dependence is similar.

The self-consistency properties of NA and NOD are shown in  \cite[Property P$_6$]{JP83} and  \cite[Lemma 1]{L66}, respectively. We will show the properties for NRD and WNA. 
Let $\mathbf Y=(f_1(X_1), \dots, f_d(X_d))$.
\begin{enumerate}

\item  Assume $\mathbf X$ is NRD. Fix $i\in [d]$. Let $g$ be a coordinatewise increasing function and $g'=g \circ(f_1, \dots, f_{i-1},f_{i+1},\dots, f_d)$. As a result,  we have $g'$ is a coordinatewise increasing function and $g(\mathbf Y_{-i})=g'(\mathbf X_{-i})$. For any $y\in \R$, let  $A_y=\{x: f_i(x)=y\}$. We have $\{Y_i=y\}=\{X_i\in A_y\}$.  Therefore, $\E[g(\mathbf Y_{-i})|Y_i=y]=\E[g'(\mathbf X_{-i})|X_i \in A_{y}]$.
Assume $y_1<y_2$. For any $x_1\in A_{y_1}$ and $x_2 \in A_{y_2}$, we have $x_1\le x_2$; hence,  $\E[g'(\mathbf X_{-i})|X_i=x_1]\ge \E[g'(\mathbf X_{-i})|X_i=x_2].$  Thus, 
\begin{align*}\E[g'(\mathbf X_{-i})|X_i \in A_{y_1}]&=\E[\E[g'(\mathbf X_{-i})|X_i]|X_i \in A_{y_1}]\\
&\ge \E[\E[g'(\mathbf X_{-i})|X_i]|X_i \in A_{y_2}]=\E[g'(\mathbf X_{-i})|X_i \in A_{y_2}],
\end{align*}
which implies that $\E[g(\mathbf Y_{-i})|Y_i=y_1]\ge \E[g(\mathbf Y_{-i})|Y_i=y_2]$; hence $\mathbf Y$ is NRD.
\item Assume $\mathbf X$ is WNA. For $i\in [d]$, let $S\subseteq \R^{d-1}$ be a decreasing set,  and $$S_i^{f}=\{ (x_1, \dots x_{i-1}, x_{i+1}, \dots, x_d): (f_1(x_1),\dots, f_{i-1}(x_{i-1}), f_{i+1}(x_{i+1}), \dots, f_d(x_d)) \in S \}.$$
It is clear that $\{\mathbf Y_{-i} \in S\}=\{\mathbf X_{-i} \in S^{f}_i\}$. For any 
 ${\mathbf x}_1 \le {\mathbf x}_2$ and ${\mathbf x}_2 \in S^{f}_i$, we have  $f_k(x_{1,k})\le f_k(x_{2,k})$ for all $k\in [d]\setminus \{i\}$. Furthermore, because $S$ is decreasing, we have ${\mathbf x}_1 \in S^{f}_i$, which implies $S^{f}_i$ is  a decreasing set. 
For any $y\in \R$ with $\p(Y_i\le y)>0$, let 
$x=\sup\{t \in \R: f_i(t)\le y\}$. 
If $f_i(x)\le y$, we have 
 $\{Y_i\le y\}=\{X_i \le x\}$ and $\p(X_i\le x)>0$. Therefore,
$$\p(\mathbf Y_{-i}\in S|Y_i\le y)=\p\left(\mathbf X_{-i}\in S^{f}_i| X_i \le x\right)\le \p\left(\mathbf X_{-i}\in S^{f}_i\right)=\p(\mathbf Y_{-i}\in S),$$
which implies that $\mathbf Y$ is WNA. If $f_i(x)>y$, we have 
 $\{Y_i\le y\}=\{X_i < x\}$ and $\p(X_i< x)>0$. Therefore,
\begin{align*}
     \p(\mathbf Y_{-i}\in S, Y_i\le y) &=\p\left(\mathbf X_{-i}\in S^{f}_i, X_i <x\right)\\
    &=\lim_{t \uparrow x} \p\left(\mathbf X_{-i}\in S^{f}_i, X_i \le t\right)\\
    &\le \lim_{t \uparrow x} \p\left(\mathbf X_{-i}\in S^{f}_i\right)\p( X_i \le t)\\
    &=\p\left(\mathbf X_{-i}\in S^{f}_i\right) \lim_{t \uparrow x}\p( X_i \le t)\\
    &=\p\left(\mathbf X_{-i}\in S^{f}_i\right)\p( X_i <x )=\p(\mathbf Y_{-i}\in S)\p( Y_i\le y),
\end{align*}
which implies that $\p(\mathbf Y_{-i}\in S| Y_i\le y)\le \p(\mathbf Y_{-i}\in S)$ and $\mathbf Y$ is WNA.\qedhere
\end{enumerate}
\end{proof}

The following theorem demonstrates that the checkerboard copula of $\X$ preserves the dependence information of $\X$. 
\begin{theorem}\label{th:gen}
A random vector $\mathbf X$ satisfies $\mathfrak D$ if and only if it has a copula that satisfies $\mathfrak D$. Moreover, the copula can be chosen as the checkerboard copula $C_{\mathbf X}^{\perp}$.
\end{theorem}

\begin{proof}

 The ``if" part follows from Lemma \ref{lem:2} because, for $\mathbf U=(U_1,\dots,U_d)$ following the copula of $\X$ that satisfies $\mathfrak D$, we have  $(X_1, \dots, X_d)=\left(F^{-1}_1(U_1), \dots, F^{-1}_d(U_d)\right)$ and $F^{-1}_i$ is increasing for all $i\in[d]$.

Now we show the ``only if" part.  Let $\mathbf U=(U_1,\dots,U_d)$ be the random vector given by \eqref{eq:us} with $\mathbf V = (V_1,\dots,V_d) \sim \mathrm{U}\left([0,1]^d\right)$ independent of $\mathbf X$. 
 Hence, we have $\mathbf U \sim C_{\mathbf X}^{\perp}$ and $C_{\mathbf X}^{\perp}$ is a copula of $\mathbf X$. Note that, for any $i\in [d]$, given $V_i$, we have that $U_i$ is an increasing function of $X_i$. Hence, by Lemma \ref{lem:2}, $\mathbf X$ satisfies $\mathfrak D$ implies that $\mathbf  U|\mathbf  V$ also satisfies $\mathfrak D$.

Assume $\mathbf X$ is NA. 
For any given pair of disjoint subsets $A$, $B$ of $[d]$ and any given functions  $f$ and $g$ both increasing or decreasing coordinatewise, we have
\begin{align*}
 \cov(f(\mathbf U_{A}), g(\mathbf U_{B}))
&=\E[\cov(f(\mathbf U_{A}), g(\mathbf U_{B})|\mathbf V)]+\cov\left(\E[f(\mathbf U_{A})|\mathbf V], \E[g(\mathbf U_{B})|\mathbf V]\right)\\
&\le 0+\cov\left(\E[f(\mathbf 
 U_{A})|\mathbf V_{A}], \E[g(\mathbf U_{B})|\mathbf V_{B}]\right)=0,
\end{align*}
where the inequality follows from $\mathbf  U|\mathbf  V$ is NA, and the last equality follows from the independence between $\mathbf V_{A}$ and $\mathbf V_{B}$. Hence, $\mathbf U$ is NA.

Assume $\mathbf X$ is NRD. For any fixed $i$ and $k$, by \eqref{eq:us}, there exist $x$ and $v$ such that $\{ U_i = k\} = \{ X_i = x, V_i = v \}$. Then, for any coordinatewise increasing function $g$, by the independence between $V_i$ and $(X_i, \mathbf{U}_{-i})$, we have
\begin{equation*}
    \E[g(\mathbf U_{-i})|U_i=k]= \E[g(\mathbf U_{-i})| X_i = x, V_i = v ] =  \E[g(\mathbf U_{-i})| X_i = x] . 
\end{equation*}
Because $\mathbf{U}_{-i}$ is a function of $\mathbf{X}_{-i}$ and $\mathbf{V}_{-i}$, we can let $h$ be the function such that $g(\mathbf{U}_{-i}) = h(\mathbf{X}_{-i}, \mathbf{V}_{-i})$. Then, due to the independence between $\mathbf{V}_{-i}$ and $\mathbf{X}$, 
\begin{equation*}
    \E[g(\mathbf U_{-i})| X_i = x] = \E[h(\mathbf{X}_{-i}, \mathbf{V}_{-i})| X_i = x] = \int_{[0,1]^{d-1}}  \E[h(\mathbf{X}_{-i}, \mathbf{v}_{-i}) | X_i=x] \mathrm{d}\mathbf{v}_{-i},
\end{equation*}
where $\mathbf{v}_{-i} = (v_1, \dots, v_{i-1},v_{i+1},\dots,v_d)$. Therefore, for any $k_1 \leq k_2$, there exist $x_1$ and $x_2$ such that
\begin{align*}
    \E[g(\mathbf U_{-i})|U_i=k_1] &= \int_{[0,1]^{d-1}}  \E[h(\mathbf{X}_{-i}, \mathbf{v}_{-i}) | X_i=x_1] \mathrm{d}\mathbf{v}_{-i} , \\
    \E[g(\mathbf U_{-i})|U_i=k_2] &= \int_{[0,1]^{d-1}}  \E[h(\mathbf{X}_{-i}, \mathbf{v}_{-i}) | X_i=x_2] \mathrm{d}\mathbf{v}_{-i} .
\end{align*}
In addition, by \eqref{eq:us}, we must have $x_1 \leq x_2$. Note that given $\mathbf{v}_{-i}$,  $h(\mathbf{X}_{-i}, \mathbf{v}_{-i})$ is a coordinatewise increasing function of $\mathbf{X}_{-i}$. Hence, we have $\E[h(\mathbf{X}_{-i}, \mathbf{v}_{-i}) | X_i=x_1] \geq \E[h(\mathbf{X}_{-i}, \mathbf{v}_{-i}) | X_i=x_2]$ for any $\mathbf{v}_{-i}$. Therefore, $\E[g(\mathbf U_{-i})|U_i=k_1] \geq \E[g(\mathbf U_{-i})|U_i=k_2]$ and $\mathbf{U}$ is NRD. 

Assume $\mathbf X$ is WNA.
 For any $i\in[d]$,  decreasing set $S  \subseteq  \R^{d-1}$, and $x\in \R$ with $\p(U_i\le x)>0$,  
\begin{align*}
 \p(\mathbf U_{-i} \in S, U_i\le  x) &= \E[\p(\mathbf U_{-i} \in S,  U_i\le  x \mid \mathbf V)]\\
 &\le \E[\p(\mathbf U_{-i} \in S|\mathbf V_{-i})\p(  U_i\le  x \mid V_i)]\\
 &=\E[\p(\mathbf U_{-i} \in S|\mathbf V_{-i})]\E[\p(  U_i\le  x \mid V_i)]\\
&=\p(\mathbf U_{-i} \in S)\p(  U_i\le  x).
 \end{align*}  
Hence, $\mathbf U$ is WNA.

Assume $\mathbf X$ is NOD. For any $t_1, \dots, t_d \in \R$, we have
\begin{align*}
\p(U_1\le t_1, \dots, U_d \le t_d)&=\E[\p(U_1\le t_1, \dots, U_d \le t_d|V_1, \dots, V_d)]\\
&\le \E[\p(U_1\le t_1|V_1) \cdots\p(U_d \le t_d| V_d)]\\
&= \E[\p(U_1\le t_1|V_1)] \cdots \E[\p(U_d \le t_d|V_d)]\\
&=\p(U_1\le t_1) \cdots\p(U_d \le t_d).
\end{align*}
Similarly, we can show
$$
\p(U_1 > t_1, \dots, U_d > t_d) \le \p(U_1> t_1) \cdots\p(U_d > t_d).
$$
Hence, $\mathbf U$ is NOD.

In conclusion, if $\mathbf X$ satisfies $\mathfrak D$, then $\mathbf U$ satisfies $\mathfrak D$, where $\mathfrak D$ is one of the four concepts of negative dependence.

To show the case of positive dependence, we follow a similar route. 
We take the same $\mathbf U$ as above.  
Assume $\mathbf X$ is PA. Because $U_i |V_i$ is an increasing function of $X_i$, by Lemma \ref{lem:2}, $\mathbf U|\mathbf V$ is also PA. Thus, for any given pair of  subsets $A,B$ of $[d]$ and any given functions  $f$ and $g$ both coordinatewise increasing or decreasing, we have
\begin{align*}
\cov(f(\mathbf U_{A}), g(\mathbf U_{B}))
&=\E[\cov(f(\mathbf U_{A}), g(\mathbf U_{B})|\mathbf V)]+\cov\left(\E[f(\mathbf U_{A})|\mathbf V], \E[g(\mathbf U_{B})|\mathbf V]\right)\\
&\ge \cov\left(\E[f(\mathbf 
 U_{A})|\mathbf V_{A}], \E[g(\mathbf U_{B})|\mathbf V_{B}]\right).
\end{align*}
Moreover, given $\mathbf X$, $U_i$ is an increasing function of $V_i$. Hence, $\E[f(\mathbf U_{A})|\mathbf V_{A}]$ and $\E[f(\mathbf U_{B})|\mathbf V_{B}]$ are coordinatewise increasing (or decreasing) with respect to $\mathbf V_A$ and $\mathbf V_B$, respectively, if $f$ and $g$ are both coordinatewise increasing (or decreasing). Because $\mathbf V$ is PA, we have
$$
\cov\left(\E[f(\mathbf 
 U_{A})|\mathbf V_{A}], \E[g(\mathbf U_{B})|\mathbf V_{B}]\right) \geq 0,
$$
implying that $\mathbf U$ is PA. The proofs for other positive dependence concepts are similar. {A partial proof for POD can also be found in \citet[Proposition 2.3]{DFQU15}.}
\end{proof}

{
\cite{GN07} also show the dependence preservation results for positive orthant dependence, positive likelihood ratio dependence, and tail dependence in bivariate case.}
{\begin{remark}
We can clearly see from Theorem \ref{th:gen} that $C^\perp_{\mathbf X}=\Pi$ (the independence copula) is the only independent copula in $\mathcal{C}_{\mathbf X}$ when $\mathbf X$ is independent. This fact is used as the basis of the independence test in \cite{GNRM19}.
\end{remark}}
\section{{Two consequences of Theorem 3}}
\label{sec:app}
 
We provide two applications in this section to highlight the usefulness of Theorem \ref{th:gen}.
 
\subsection{Diversification penalty}\label{sec:cew}
 
For random variables $X$ and $Y$, let $X\ge_{\rm st} Y$ represent $\p(X>x)\ge \p(Y>x)$ for all $x\in \R$; this is called the stochastic order.
  \cite{CEW24a, CEW24} studied the problem of  diversification penalty;
that is, whether 
\begin{align}
\label{eq:R2-new}
X\le_{\rm st} \sum_{i=1}^d \theta_i X_i ~\mbox{for all $(\theta_1,\dots,\theta_d)\in \Delta_d$, where $X,X_1,\dots,X_d$ are identically distributed,}
\end{align} 
holds under certain marginal distributions and dependence structures. Here, $\Delta_d$ is the standard simplex defined by $\Delta_d = \{(\theta_1,\dots,\theta_d) \in [0,1]^d: \theta_1+\dots+\theta_d=1 \}$. 
When $X$ is interpreted as a loss,  \eqref{eq:R2-new} intuitively means that the non-diversified portfolio $X$ is less dangerous than the diversified portfolio $\sum_{i=1}^d \theta_i X_i $.
This seems counter-intuitive at first glance, but it indeed happens in the model of
  \cite{CEW24a}, where  $X$ has infinite mean.

Define the set, for some dependence concept $\mathfrak D$ in Section \ref{sec:51}, 
$$
\mathcal F_{\mathfrak D}=\{\mbox{distribution of $X$}:\mbox{\eqref{eq:R2-new} holds for all $(X_1,\dots,X_d)$ that satisfy  $\mathfrak D$}\}.
$$ 
  \cite{CEW24a} showed that the Pareto(1) distribution belongs to $\mathcal F_{\mathrm{WNA}}$, and hence also to $\mathcal F_{\mathrm{NA}}$, $\mathcal F_{\mathrm{NRD}}$, and  $ \mathcal F_{\mathrm{IN}} $, where IN stands for independence.
Moreover, \cite[Proposition 1]{CEW24} showed that 
$\mathcal F_{\mathfrak D}$  for $\mathfrak D$ being WNA, NA, or IN  is closed under strictly increasing convex transforms on the random variables. 
Our next result, which relies on our Theorem \ref{th:gen},
addresses non-strictly increasing $f$ and other notions of dependence, thus generalizing the above result.

\begin{proposition}
\label{prop:cew} Each of $\mathcal F_{\mathfrak D}$ is closed under increasing convex transforms on the random variable. 
\end{proposition}
\begin{proof}
Below we first show that each of $\mathcal F_{\mathfrak D}$ is closed under strictly increasing convex transforms on the random variable, that is,  if the distribution of $X$ is in $\mathcal F_{\mathfrak D}$, so is the distribution of $f(X)$ for a strictly increasing convex $f$.  
Assume that $F\in \mathcal F_{\mathfrak D}$, $X$ follows $F$,
and $Y=f(X)$, where $f$ is strictly increasing and convex. 
Because $f$ is strictly increasing, if $(Y_1,\dots,Y_d)$ satisfies $\mathfrak D$, so does $(X_1,\dots,X_d)$, where $X_i=f^{-1}(Y_i)$ for $i\in [d]$,  by Lemma \ref{lem:2}. 
Because each of $X,X_1,\dots,X_d$ has a distribution $F\in \mathcal F_{\mathfrak D}$, we have $X\le_{\rm st} \sum_{i=1}^d \theta_i X_i$,
and this gives, using the convexity of $f$, 
\begin{equation}\label{eq:cew2} Y=  f(X)\le_{\rm st}f\left( \sum_{i=1}^d \theta_i X_i\right) \le \sum_{i=1}^d \theta_i  f (X_i) = \sum_{i=1}^d \theta_i Y_i  .\end{equation}

To address the case that $f$ is not strictly increasing, Theorem \ref{th:gen} allows us to find the above $(X_1,\dots,X_d)$ that satisfies $\mathfrak D$ and such that $Y_i=f(X_i)$ for $i\in [d]$.
In particular, using Theorem \ref{th:gen}, we can construct $(U_1,\dots,U_d)$ that follows the checkerboard copula of $(Y_1,\dots,Y_d)$ and satisfies $\mathfrak D$, such that 
  $$
  (Y_1,\dots,Y_d)  = (f\circ g (U_1),\dots,f\circ g(U_d)),
  $$
   where 
    $g$ is the quantile function of $X$ and $f\circ g$ is the quantile function of $Y$.
    Setting $(X_1,\dots,X_d)= (g(U_1),\dots,g(U_d))$, we get that  $(X_1,\dots,X_d)$ satisfies $\mathfrak D$, and this leads to \eqref{eq:cew2}. 
\end{proof}

\subsection{Induced order statistics}\label{sec:ios}

Here we demonstrate another application of Theorem \ref{th:gen} in characterizing the distribution of induced order statistics. Consider $N$ independent and identically distributed bivariate random vectors
\begin{equation*}
        \begin{pmatrix} \xi_{1} \\ \eta_1 \end{pmatrix}, \begin{pmatrix} \xi_{2} \\ \eta_2 \end{pmatrix} , \dots, \begin{pmatrix} \xi_{N} \\ \eta_N \end{pmatrix}.
\end{equation*}
Note that, for $i\neq j$, $(\xi_i, \eta_i)$ and $(\xi_j, \eta_j)$ are independent and identically distributed, but $\xi_i$ and $\eta_i$ may be correlated and have different marginal distributions. We rank these bivariate vectors according to their first components, $\xi_i$:
\begin{equation}
\label{equ:ios}
    \begin{pmatrix} \xi_{1:N} \\ \eta_{[1:N]} \end{pmatrix}, \begin{pmatrix} \xi_{2:N} \\ \eta_{[2:N]} \end{pmatrix} , \dots, \begin{pmatrix} \xi_{N:N} \\ \eta_{[N:N]} \end{pmatrix},
\end{equation}
where $\xi_{1:N} \leq \xi_{2:N} \leq \cdots \leq \xi_{N:N}$ are the order statistics of $\xi_1, \xi_2, \dots, \xi_N$. The notation $\eta_{[i:N]}$ represents the $i$-th \textit{induced order statistic} \citep{bhattacharya:1974}, where the order is induced by another variable $\xi_i$. The induced order statistics $\eta_{[1:N]},\dots,\eta_{[N:N]}$ are also referred to as \textit{concomitants} of the order statistics $\xi_{1:N},\dots,\xi_{N:N}$ \citep{david:1973}.

In the context of constructing impact portfolios, \cite{LWZZ22} investigated the joint distribution of $(\eta_{[1:N]}, \dots, \eta_{[N:N]})$. In particular, they proved a representation theorem for the joint distribution of $(\eta_{[1:N]}, \dots, \eta_{[N:N]})$ using the copula of $(\xi_i,\eta_i)$. Furthermore, they demonstrated that if $\xi_i$ is not continuously distributed, the representation theorem holds if and only if the copula of $(\xi_i,\eta_i)$ is chosen as the (bivariate) checkerboard copula in this paper. This reveals a potential application of the checkerboard copula in portfolio construction. 

\cite{LWZZ22} also showed that the rank of the odd-order moments of induced order statistics relies on the copula of $(\xi_i,\eta_i)$. Assume that $C$ is a copula of $(\xi_i,\eta_i)$. \cite[Theorem EC.5]{LWZZ22} proved that, for any $k=0,1,\dots$, if $C$ is PRD, we have
    \begin{equation}\label{equ:ios_increasing}
        \mathbb{E} \left( \eta_{[1:N]}^{2k+1} \right) \leq \mathbb{E}\left(\eta_{[2:N]}^{2k+1}\right) \leq \cdots \leq \mathbb{E}\left(\eta_{[N:N]}^{2k+1}\right),
    \end{equation}
    and if $C$ is NRD, we have
    \begin{equation}\label{equ:ios_decreasing}
        \mathbb{E}\left(\eta_{[1:N]}^{2k+1}\right) \geq \mathbb{E}\left(\eta_{[2:N]}^{2k+1}\right) \geq \cdots \geq \mathbb{E}\left(\eta_{[N:N]}^{2k+1}\right).
    \end{equation}
In particular, the copula $C$ can be chosen as the checkerboard copula. Therefore, using Theorem \ref{th:gen}, we directly obtain the following result. 

\begin{proposition}
\label{prop:ios} For any $k=0,1,\dots$, \eqref{equ:ios_increasing} holds if $(\xi_i,\eta_i)$ is PRD, and \eqref{equ:ios_decreasing} holds if $(\xi_i,\eta_i)$ is NRD. 
\end{proposition}

The difference between Proposition \ref{prop:ios} and \cite[Theorem EC.5]{LWZZ22} is that the latter imposes the dependence assumption (PRD or NRD) on the copula of $(\xi_i,\eta_i)$, while the former imposes a more natural assumption on the random vector $(\xi_i,\eta_i)$ directly, which is only possible due to our Theorem \ref{th:gen}.



{\section{Applications in co-risk measures and portfolio selection}\label{sec:application_corisk}

In this section, we use both numerical and empirical experiments to show that the choice of copula impacts the calculation of co-risk measures when the marginal distributions are not continuous. In particular, we consider  Marginal ES as defined in \eqref{equ:MES}, which is discussed in our third motivating question in Section \ref{sec:1}. Section \ref{subsec:numerical} presents a numerical experiment to show that different copula choices can lead to varying Marginal ES results. Section \ref{subsec:empirical} uses real data from the U.S. stock market to illustrate how the choice of copula affects the performance of the portfolio with minimum Marginal ES. Our results not only highlight the importance of copula selection in the computation of co-risk measures in financial practice, but also show that the checkerboard copula is often the most convenient and natural choice that can produce reliable results.

\subsection{Numerical experiment}\label{subsec:numerical}

Consider a bivariate normal distribution with marginals
$\mathrm{N}(0,\sigma^2)$ and correlation $r$. Denote this bivariate distribution by $F_r$. We choose this distribution because it is well known that, for $(X_1, X_2) \sim F_r$, the Marginal ES of $X_2$ given $X_1$ at level $p \in (0,1)$ can be explicitly computed as
\begin{equation}\label{equ:MES_normal}
    \rho(X_2|X_1) = \mathbb{E}[X_2 | X_1 > \Phi^{-1}(p)] = \frac{r \sigma}{1-p} \varphi\left( \Phi^{-1}(p) \right),
\end{equation}
where $\varphi$ and $\Phi$ are the density and distribution functions of $\mathrm{N}(0,1)$, respectively.

Now we conduct a numerical experiment based on this bivariate normal distribution. We draw 1,000 bivariate random vectors, $\left(X_1^{(1)},X_2^{(1)}\right), \dots, \left(X_1^{(1,000)},X_2^{(1,000)}\right)$, from $F_r$. These random values are then rounded to one decimal place to estimate an empirical bivariate distribution, denoted by $\hat{F}$. Thus, $\hat{F}$ is a discrete bivariate distribution on the discrete grid caused by rounding. Next, assuming that $(\hat{X}_1, \hat{X}_2) \sim \hat{F}$, we compute the Marginal ES of $\hat{X}_2$ given $\hat{X}_1$ at level $p \in (0,1)$, which is $\rho(\hat{X}_2|\hat{X}_1)$ defined as \eqref{equ:MES}. When calculating the Marginal ES, the following two copulas of $(\hat{X}_1, \hat{X}_2)$ given by \eqref{eq:us} are considered:
\begin{enumerate}[(i)]
    \item Let $(V_1,V_2) \sim \mathrm{U}\left([0,1]^2 \right)$ be independent of $(\hat{X}_1, \hat{X}_2)$. That is to use the checkerboard copula $C^\perp$. 
    \item Let $V_2 \sim \mathrm{U}[0,1]$ be independent of $\hat{X}_2$. In addition, given $\hat{X}_2$, let $V_2$ and $\hat{X}_1$ be comonotonic. Let $V_1 \sim \mathrm{U}[0,1]$ be independent of $\hat{X}_1$, $\hat{X}_2$, and $V_2$. We denote this copula by $C^\mathrm{+}$. 
\end{enumerate}
Therefore, we obtain different values of $\rho(\hat{X}_2|\hat{X}_1)$ using the two different copulas. 

Given $(r,p)$, we run the simulation procedure described above 10,000 times and calculate the Marginal ES under both copulas for each run. We choose $\sigma=10$. Table \ref{tab:simulation} shows the Marginal ES given by \eqref{equ:MES_normal} (normal formula), the average Marginal ES value, and the mean squared error (MSE)---the mean squared difference between $\rho(\hat{X}_2|\hat{X}_1)$ and the value in \eqref{equ:MES_normal}---across the 10,000 runs for each of the two copulas. 

\begin{table}[t]
\centering
\caption{Simulation results for computing Marginal ES under two copulas.\label{tab:simulation}}
\begin{tabular}{c|ccc|ccc|ccc}
\toprule
\multicolumn{1}{l}{} & \multicolumn{3}{|c}{$p=0.9$} & \multicolumn{3}{|c}{$p=0.95$} & \multicolumn{3}{|c}{$p=0.975$} \\\midrule
$\rho$               & 0.2     & 0.3     & 0.4     & 0.2      & 0.3     & 0.4     & 0.2      & 0.3      & 0.4     \\\midrule
normal formula        & 3.510   & 5.265   & 7.020   & 4.125    & 6.188   & 8.251   & 4.676    & 7.013    & 9.351   \\
average with $C^\perp$  & 3.502   & 5.257   & 7.001   & 4.116    & 6.171   & 8.234   & 4.688    & 6.992    & 9.330   \\
average with $C^+$      & 3.985   & 5.730   & 7.455   & 4.687    & 6.726   & 8.767   & 5.338    & 7.615    & 9.934   \\
MSE with $C^\perp$  & 0.940   & 0.920   & 0.891   & 1.914    & 1.796   & 1.673   & 3.725    & 3.636    & 3.329   \\
MSE  with $C^+$      & 1.225   & 1.192   & 1.131   & 2.348    & 2.173   & 2.020   & 4.357    & 4.187    & 3.796   \\\bottomrule
\end{tabular}
\end{table}

Table \ref{tab:simulation} illustrates that the choice of copula affects the calculation of Marginal ES. Under our setup, the Marginal ES computed using the checkerboard copula is, on average, closer to the result obtained using the normal formula \eqref{equ:MES_normal}. This demonstrates that the checkerboard copula is a good candidate for computing co-risk measures for non-continuous random variables.

\subsection{Empirical study}\label{subsec:empirical}
To further demonstrate the importance of copula selection when computing co-risk measures in financial practice, we use real stock data to calculate the Marginal ES. We obtain daily returns of the S\&P 500 Index and five widely traded US stocks---Microsoft, Apple, Google, Nvidia, and Amazon---from 2005 to 2023.\footnote{Our data comes from the Center for Research in Security Prices (CRSP). We use data starting in 2005 because Google became publicly traded in August 2004.} To make the distribution discrete, we classify market conditions into five groups based on the daily returns of the S\&P 500 Index: $(-\infty, -3\%]$, $(-3\%, -1\%]$, $(-1\%, +1\%]$, $(+1\%, +3\%]$, and $(+3\%, +\infty)$. These five conditions represent very bad, bad, fair, good, and very good, with corresponding values of $-2$, $-1$, $0$, $+1$, and $+2$, respectively. 

Next, for each individual stock, we use the negative values of its daily returns along with these market condition values over the entire period to estimate an empirical bivariate distribution. Then, based on the empirical bivariate distribution, the two copulas described in Section \ref{subsec:numerical} are applied to calculate the Marginal ES of the stock's loss given the market condition. We also present results computed using the empirical bivariate distributions of daily stock returns and S\&P 500 Index returns directly, without classifying the market conditions.

Table \ref{tab:empirical_MES} shows the Marginal ES of the five stocks under the two copulas, along with the result computed directly from the empirical distributions. The Marginal ES values differ across the three methods, and the results from the checkerboard copula, $C^\perp$, are lower than those from the alternative copula, $C^+$.
\begin{table}[t]
\centering
\caption{Marginal ES of the loss of five individual stocks given the loss of the S\&P 500 Index.\label{tab:empirical_MES}}
\begin{tabular}{cccccc}
\toprule
\multicolumn{1}{l}{} & \multicolumn{1}{c}{Microsoft} & \multicolumn{1}{c}{Apple} & \multicolumn{1}{c}{Google} & \multicolumn{1}{c}{Nivdia} & \multicolumn{1}{c}{Amazon} \\\midrule
$C^\perp$  & 3.518\%                       & 3.545\%                   & 3.455\%                    & 5.043\%                    & 3.712\%                    \\
$C^+$      & 4.404\%                       & 4.752\%                   & 4.523\%                    & 6.686\%                    & 4.998\%                   \\
empirical distribution & 3.845\% & 3.905\% & 3.735\% & 5.456\% & 4.054\% \\\bottomrule
\end{tabular}
\end{table}

Different choices of copulas can also result in varying financial performance for the portfolio with minimum Marginal ES (minMES). To demonstrate this, we construct minMES portfolios for the five stocks as follows. At the beginning of year $t$, we determine the weights of the five stocks that minimize the Marginal ES of the portfolio's loss given the market condition, using data from year $t-1$. The optimization is subject to the constraints that all weights must be non-negative and sum to 1. These optimal weights are then held throughout year $t$.

Figure \ref{fig:empirical_performance} and Table \ref{tab:empirical_metric} show the cumulative portfolio values and performance metrics of this minMES strategy for $p=0.975$ under the three methods over the entire sample period, respectively. We find that the choice of copula significantly impacts the financial performance of the minMES portfolio. Furthermore, in our empirical study, the checkerboard copula generally achieves better performance, demonstrating that it is a convenient and effective choice for producing reliable results for the considered dataset.
Certainly, we do not claim that this advantage is profitable in general, which requires comprehensive empirical analysis.

\begin{figure}[t]
    \centering
        \includegraphics[width=0.80\linewidth]{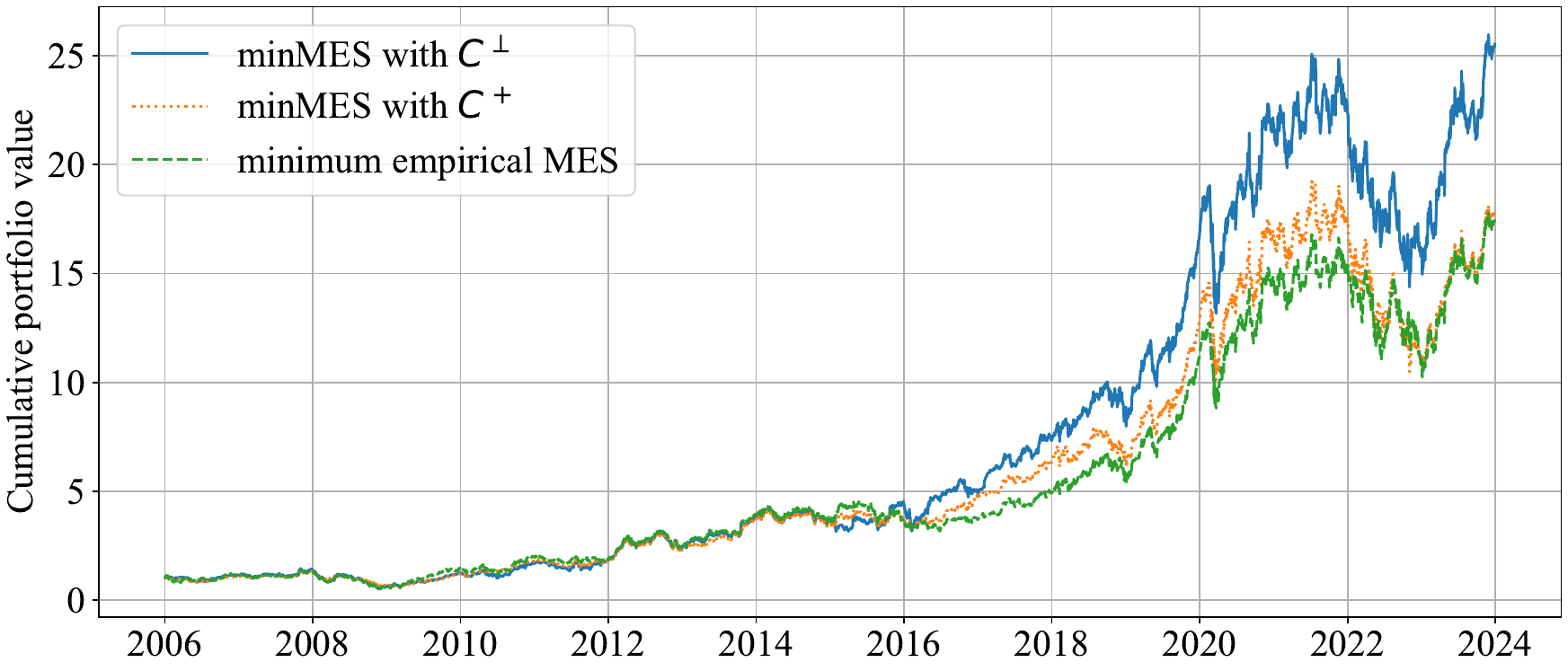}
    \caption{Cumulative values of the minMES portfolios.\label{fig:empirical_performance}}
\end{figure}

\begin{table}[t]
\centering
\caption{Performance metrics (average annual returns, standard deviation of annual returns, and Sharpe ratio) for the minMES portfolios. The Sharpe ratio is calculated assuming a risk-free rate of $3\%$.\label{tab:empirical_metric}}
\begin{tabular}{cccc}
\toprule
metrics  & \multicolumn{1}{c}{average return} & \multicolumn{1}{c}{standard deviation} & \multicolumn{1}{c}{Sharpe ratio} \\\midrule
$C^\perp$            & 22.34\%                            & 29.89\%                                & 0.6470                                            \\
$C^+$                & 19.73\%                            & 27.65\%                                & 0.6050                                   \\
empirical distribution & 20.11\% & 29.92\% & 0.5721 \\
\bottomrule
\end{tabular}
\end{table}
}

\section{Conclusion} 
\label{sec:conclusion}

We discussed the choice of copula when the marginal distributions are not necessarily continuous. Among all the choices of copulas for a given random vector, the checkerboard copula is the most convenient and natural selection in applications such as simulating from the copula, stressing the distribution, and computing a co-risk measure. 
It is shown that the checkerboard copula is the most unbiased choice in the sense that it has the largest Shannon entropy among all possible copulas for a given random vector. Moreover, the checkerboard copula can preserve the dependence information of the underlying random vector. This preservation property is applied to 
identify suitable distributions 
in the context of diversification penalty studied by \cite{CEW24a, CEW24}
and to 
determine the ranks of the moments of induced order statistics in the context of impact portfolios studied by \cite{LWZZ22}. {Finally,
our results indicate that the choice of copula significantly affects the calculation of co-risk measures when the marginal distributions are not continuous. Through numerical experiments and empirical studies, we find that the checkerboard copula can produce reliable results when computing Marginal ES and demonstrate strong performance when constructing minMES portfolios. }

\end{document}